\documentclass[12pt]{article}
\usepackage{amsmath}
\usepackage{graphicx,psfrag,epsf}
\usepackage{enumerate}
\usepackage{psfrag,epsf}
\usepackage{url} % not crucial - just used below for the URL
\usepackage{amsmath,amssymb,amsfonts,amsthm,mathtools,mathrsfs}
\usepackage{array}
\usepackage{makecell}
\usepackage{tikz}
\usepackage{pgfplots}
\usepgfplotslibrary{fillbetween}
\pgfplotsset{compat=1.16}

\usepackage[normalem]{ulem}%\sout{riscar}

\usepackage{bm,bbm}
\usepackage{color}
\usepackage{subfigure}
\usepackage{algorithm,algpseudocode}
\usepackage[symbol]{footmisc}
\usepackage{scalefnt}
\usepackage{authblk} % author and affiliations
\usepackage{multirow,centernot}
\usepackage[colorlinks=true, citecolor=blue, urlcolor=blue]{hyperref}
\usepackage{listings}
\usepackage{xcolor}

\usepackage{natbib}
\usepackage{dsfont}
\usepackage[title]{appendix}
\usepackage{newtxtext,newtxmath} % Times Roman fonts (text and math)

\input cyracc.def

\usepackage[left=1in,top=1.1in,right=0.5in,bottom=1in]{geometry}

\theoremstyle{definition}

\newtheorem{theorem}{Theorem}[section]

\newtheorem{definition}{Definition}[section]

\newtheorem{proposition}[theorem]{Proposition}
\newtheorem{remark}[theorem]{Remark}
\makeatletter
\def\@seccntformat#1{\@ifundefined{#1@cntformat}%
	{\csname the#1\endcsname\quad}%      default
	{\csname #1@cntformat\endcsname}%    enable individual control
}
\makeatother

\newif\ifShowComments
\ShowCommentstrue
\def\strutdepth{\dp\strutbox}
\def\druk#1{\strut\vadjust{\kern-\strutdepth
		{\vtop to \strutdepth{%
				\baselineskip\strutdepth\vss
				\llap{\hbox{#1}\quad}\null}}}}

\title{\bf
Unbiased estimation of normalized scale-invariant indices
under the {Gamma} distribution
}

\author{
\text{Roberto Vila}$^{1}$\thanks{Corresponding author: Roberto Vila, email: {rovig161@gmail.com}
}\, ,
\, 
\text{Helton Saulo}$^{1,2}$ 
\,\,and
\text{Felipe Quintino}$^{1}$
\\
{\small $^{1}$ Department of Statistics, University of Brasilia, Brasilia, Brazil}\\
{\small $^{2}$ Department of Economics, Federal University of Pelotas, Pelotas, Brazil}\\
}
\begin{document}
	\maketitle 	
	\begin{abstract} 
		We introduce a broad class of normalized scale-invariant indices ({NSII}s) generated by 
		%homogeneous functions 
		{measurable symmetric} homogeneous functions
		and encompassing several well-known measures, including the Gini coefficient, generalized Gini indices, entropy-based measures, and variability indices. Explicit expressions are obtained for these indices under {Gamma} populations. Exploiting the independence between the total sum and the associated Dirichlet proportions, we derive a simple unbiased estimator based on a U-statistic. The resulting estimator is shown to be unbiased for any NSII
		{  in the proposed class}
		when the underlying population follows a {Gamma} distribution. 
		Several examples are provided to illustrate the general theory.
		A Monte Carlo simulation study is carried out that shows the good performance of the unbiased estimator in several scenarios of index choices.
		%We also present a simulation study that goes beyond the established theory by examining the estimator’s applicability in settings characterized by a generalized {Gamma} distribution.
		We evaluate the effectiveness of the {NSII}s and their estimates in modeling a real-world dataset related to gross domestic product per capita in the Americas.
		
	\end{abstract}
	
	%	
	
	%\clearpage
	%{
		%	\hypersetup{linkcolor=black}
		%	\tableofcontents
		%}

	\section{Introduction}
	
	Measures of inequality, variability, and concentration play a fundamental role in economics, finance, reliability, actuarial science, and many other applied fields. Among them, the Gini coefficient \citep{Gini1912,Gini1936} is perhaps the most widely used due to its intuitive interpretation and desirable theoretical properties. Over the years, numerous extensions and related measures have been proposed, including generalized Gini indices, entropy-based measures, and order-statistic functionals; see, for example, \cite{Atkinson1970,Cowell2011,Yitzhaki2013}.
	
	Although these indices are often introduced separately, many share common structural features, such as scale invariance and dependence on relative proportions rather than absolute magnitudes. This observation motivates the search for a unified framework capable of encompassing a broad collection of existing measures while facilitating the derivation of their statistical properties.
	
	In this paper, we introduce a class of normalized scale-invariant indices ({NSII}s) generated by homogeneous functions. This framework includes, as special cases, the classical Gini coefficient, generalized Gini-type indices, entropy measures, product-type indices, and several order-statistic-based measures. We establish a general representation of these indices and derive explicit expressions when the underlying population follows a {Gamma} distribution.
	
	Typical examples of the {NSII} (listed in Table \ref{NPRI-index}) include the Gini {coefficient} \citep{Gini1912,Yitzhaki2013}, the $m$th Gini index	\citep{Gavilan-Ruiz2024,Vila2025}, the extended $m$th Gini index	\citep{Vila2026a}, the linear order-statistic inequality index	\citep{Vila2026b}, the power index \citep{Atkinson1970},	the Shannon entropy index \citep{Shannon1948,CoverThomas2006},	the product index \citep{Bullen2003}, 
	%the asymmetric index
	{ the symmetric subcase (of the asymmetric) index}
	\citep{MarshallOlkinArnold2011}, 
	%the indicator index \citep{ShakedShanthikumar2007}, 
	and the maximum and minimum indices	\citep{DavidNagaraja2003}.
	
	The {Gamma} family occupies a prominent position in probability and statistics due to its flexibility and widespread applicability. Moreover, {Gamma} random samples possess a remarkable decomposition into an independent total sum and a Dirichlet vector of proportions \citep{Mosimann1962}, which provides a natural setting for studying normalized scale-invariant indices.
	
	% Our main contribution is the analytical construction of a simple unbiased estimator for any NPRI under {\color{red}Gamma} populations. The estimator is based on a U-statistic and exploits the independence between the sample mean and the normalized proportions. As a result, a unified estimation framework is obtained for the entire NPRI class. The proposed methodology is investigated through simulation studies and an application to economic inequality data from the Americas. 
	{
		Our main contribution is the construction of a simple unbiased estimator
		for any {NSII} under Gamma populations. The estimator is based
		on a U-statistic and exploits the independence between the sample mean
		and the normalized proportions, together with the homogeneity of the
		underlying function $g$. This leads to a unified estimation framework
		for the entire {NSII} class. Several indices previously studied
		individually in our work, including the $m$th Gini, extended $m$th Gini,
		and linear order-statistic inequality indices, are recovered as special
		cases of the general result. In contrast to the index-specific
		derivations developed in those studies, the present framework treats
		arbitrary measurable { symmetric} homogeneous functions and provides a single
		argument covering all such indices. Thus, the novelty of this paper
		lies in this general formulation and unified estimation theory, rather
		than in the introduction of new individual indices. The proposed
		methodology is further investigated through simulation studies and an
		application to economic inequality data from the Americas.
	}
	
	The remainder of the paper is organized as follows. Section \ref{A class of scale-invariant pairwise indices} introduces the {NSII} framework. Section \ref{sect-3} presents the proposed estimation methodology and its theoretical properties. Section \ref{Simulation study} contains simulation results. Section \ref{sec:applications} illustrates the methodology with a real-data application. Section \ref{sect-5}  concludes the paper.

	\section{A class of normalized scale-invariant indices}\label{A class of scale-invariant pairwise indices}
	
	Let $X$ be a {\color{black} positive} random variable with finite mean 
	$\mu = \mathbb{E}[X] \in (0,\infty)$, and let $X_1,\ldots,X_m$, $m\geqslant 2$, be independent copies of $X$.
	
	% Let $g:(0,\infty)^m \to [0,\infty)$ be a measurable function satisfying:
	
	Let $g:(0,\infty)^m \to [0,\infty)$ be a measurable symmetric function{\color{black}, that is,
		$
		g(x_1,\ldots,x_m)
		=
		g(x_{\pi(1)},\ldots,x_{\pi(m)})
		$
		for every permutation $\pi$ of $\{1,\ldots,m\}$, 
	}
	satisfying:
	
	\begin{enumerate}
		\item[(A1)] $g(\lambda x_1,\ldots,\lambda x_m)=\lambda g(x_1,\ldots,x_m)$ for all $\lambda>0$ (homogeneous function of degree one);
		{
			\item[(A2)] The sharp normalizing constant
			\[
			C_g:=
			\sup_{\boldsymbol{x}\in(0,\infty)^m}
			\frac{g(x_1,\ldots,x_m)}
			{\sum_{i=1}^m x_i}
			\]
			satisfies
			$
			0<C_g<\infty.
			$
		}
	\end{enumerate}

	\begin{definition}
		The index associated with $g$ is defined by
		\begin{equation}
			\label{eq:index}
			I(g):=I(g;X)
			=
			\frac{\mathbb{E}[g(X_1,\ldots,X_m)]}{mC_{ g} \mu},
		\end{equation}
		{
			where $C_g$ is the sharp normalizing constant in (A2). }
	\end{definition}
	{
		The use of the sharp constant $C_g$ makes the normalization
		intrinsic to $g$ and, consequently, $I(g)$ is uniquely determined
		by the generating function and the distribution of $X$.
	}
	
	% We refer to $I(g)$ as a {\color{black}normalized scale-invariant} index \text{({\color{black}NSII})}, as it depends only on the relative proportion $X_1/(X_1+\cdots+X_m)$ and is normalized to take values in $[0,1]$.
	
	We refer to $I(g)$ as a { normalized scale-invariant}
	index \text{({NSII})}. 
	
	\begin{remark}
		By Assumption (A2),
		$
		0\leqslant g(X_1,\ldots,X_m)\leqslant C_{ g}\sum_{i=1}^m X_i.
		$
		Taking expectations yields
		$
		0\leqslant I(g)\leqslant 1.
		$
		
		Moreover, if \(g(x,\ldots,x)=0\) {\color{black} for all $x>0$} and \(X\) is degenerate, then \(I(g)=0\). 
		
		Finally, by the homogeneity of \(g\),
		$
		I(g;\lambda X)=I(g;X), \ \lambda>0,
		$
		that is, \(I(g)\) is scale-invariant.
	\end{remark}
	
	%The index $I(g)$ is scale-invariant, i.e., $I(g;\lambda X)=I(g;X)$ for all $\lambda>0$, and satisfies $0\leqslant I(g)\leqslant 1$. Moreover, if $X$ is degenerate, then $I(g)=0$ whenever $g(x,\ldots,x)=0$.
	
	%
	\begin{table}[htb!] 
		\centering
		\resizebox{\textwidth}{!}{
			\renewcommand{\arraystretch}{1.5}
			\setlength{\extrarowheight}{2.2pt}
			\begin{tabular}{|c|c|c|c|}
				\hline
				\textbf{Name} & ${g}(x_1,\ldots,x_m)$ & $C_{ g}$ & ${I}(g)$\\
				\hline
				
				Gini {coefficient} 
				& $|x_1-x_2|$
				& $1$ 
				& $\displaystyle \frac{\mathbb{E}|X_1-X_2|}{2\mu}$ 
				\\
				\hline
				
				The $m$th Gini index 
				& $x_{m:m}-x_{1:m}$
				& $1$ 
				& $\displaystyle \frac{\mathbb{E}[X_{m:m}-X_{1:m}]}{m\mu}$
				\\
				\hline
				
				Extended $m$th Gini index 
				& $x_{k:m}-x_{j:m}$, $1\leqslant j<k\leqslant m$
				& {$\displaystyle \frac{1}{m-k+1}$ }
				& 
				$\displaystyle \frac{
					{(m-k+1)}
					\mathbb{E}[X_{k:m}-X_{j:m}]}{m\mu}$         
				\\
				\hline
				
				Linear order-statistic inequality index 
				& \makecell{$\displaystyle \sum_{k=1}^m a_k x_{k:m}$, 
					\vspace*{0.1cm}$\displaystyle \sum_{k=1}^m a_k = 0$, \\  \vspace*{0.1cm}  $a_1 \leqslant \cdots \leqslant 
					a_m,
					\ 
					{a_m =1}
					$
				}
				& $1$ 
				& $\displaystyle
				\frac{\sum_{k=1}^m a_k\mathbb{E}[X_{k:m}]}{m\mu}$ 
				\\ 
				\hline
				
				Smoothed CV index (SCV)
				& $\displaystyle
				\frac{(x_1-x_2)^2}{x_1+x_2}$
				& $1$ 
				& $ \displaystyle
				\frac{\mathbb{E}\big[\frac{(X_1-X_2)^2}{X_1+X_2}\big]}{2\mu}$ 
				\\
				\hline
				
				Power index ($p>0$) 
				& $\displaystyle
				x_1^p(x_1+x_2)^{1-p}+x_2^p(x_1+x_2)^{1-p}$
				& $\max\{1,2^{1-p}\}$ 
				& $\displaystyle
				\frac{\mathbb{E}[g(X_1,X_2)]}{2\mu \max\{1,2^{1-p}\}}$
				\\
				\hline
				
				Shannon entropy index 
				& $\displaystyle -x_1\log\left(\frac{x_1}{x_1+x_2}\right)-x_2\log\left(\frac{x_2}{x_1+x_2}\right)$ 
				& $\log(2)$ 
				& $\displaystyle \frac{\mathbb{E}[g(X_1,X_2)]}{2\mu\log(2)}$ 
				\\
				\hline
				
				Product index 
				& $\displaystyle\frac{x_1x_2}{x_1+x_2}$
				& $\displaystyle\frac{1}{4}$ 
				& $\displaystyle \frac{2\,\mathbb{E}\!\left[\frac{X_1X_2}{X_1+X_2}\right]}{\mu}$ 
				\\
				\hline
				
				{Symmetric subcase index}
				& $\displaystyle \frac{x_1^\theta x_2^\beta}{(x_1+x_2)^{\theta+\beta-1}}$, $\theta,\beta>0$, {\color{black}$\theta=\beta$}
				& $ \displaystyle \frac{\theta^\theta\beta^\beta}{(\theta+\beta)^{\theta+\beta}}$ 
				& $ \displaystyle \frac{(\theta+\beta)^{\theta+\beta}\mathbb{E}[g(X_1,X_2)]}{
					2\theta^\theta\beta^\beta
					\mu}$ 
				\\
				\hline
				
				% Indicator index 
				% & $\displaystyle (x_1+x_2){1}_{\{\frac{x_1}{x_1+x_2}\le c\}}$, $c>0$
				% & $1$ 
				% & $\displaystyle\frac{\mathbb{E}[g(X_1,X_2)]}{2\mu}$ 
				% \\
				% \hline
				
				Maximum index  
				& {$\max\{x_1,x_2\}$}
				& $1$ 
				& $\displaystyle\frac{\mathbb{E}[\max\{X_1,X_2\}]}{2\mu}$ 
				\\
				\hline
				
				Minimum index 
				& $\displaystyle\min\{x_1,x_2\}$
				& $\displaystyle\tfrac{1}{2}$ 
				& $\displaystyle \frac{\mathbb{E}[\min\{X_1,X_2\}]}{\mu}$ 
				\\
				\hline
				
				% Linear (degenerate) 
				% & $x_1$ 
				% & $1$ 
				% & $\displaystyle\tfrac{1}{2}$ 
				% \\
				% \hline
				
			\end{tabular}
		}
		\caption{
			Examples of {NSII}s.} %The examples listed in this table are related to the Gini index \citep{Gini1912,Yitzhaki2013}, the $m$th Gini index	\citep{Gavilan-Ruiz2024,Vila2025}, the extended $m$th Gini index	\citep{Vila2026a}, the linear order-statistic inequality index	\citep{Vila2026b}, the power index \citep{Atkinson1970},	the Shannon entropy index \citep{Shannon1948,CoverThomas2006},	the product index \citep{Bullen2003}, the asymmetric index	\citep{MarshallOlkinArnold2011}, the indicator index \citep{ShakedShanthikumar2007}, and the maximum and minimum indices	\citep{DavidNagaraja2003}.}
	\label{NPRI-index}
\end{table}

% Any function $g$ satisfying (A1) and (A2) admits the representation
% \begin{equation}
	% 	\label{eq:representation}
	% 	g(x_1,\ldots,x_m) = (x_1+\cdots+x_m)\,\psi\!\left(\frac{x_1}{x_1+\cdots+x_m}\right), \quad x_1,\ldots, x_m>0,
	% \end{equation}
% for some measurable function $\psi:[0,1]\to[0,\infty)$, in which case 
% $C=\sup_{t\in[0,1]}\psi(t)$. 
% Representation \eqref{eq:representation} shows that every NPRI depends only on the relative shares
% $
% X_i/(X_1+\cdots+X_m)
% $
% rather than on the absolute scale of the observations. This explains the scale invariance of the class.
%Letting

{\color{black}
	For illustration, we derive the sharp normalizing constant $C_g$ for
	two of the indices in Table \ref{NPRI-index}: the extended $m$th Gini
	index and the linear order-statistic inequality index.
	
	For the extended $m$th Gini index,
	$
	g(x_1,\ldots,x_m)=x_{k:m}-x_{j:m}, \ 1\leqslant j<k\leqslant m,
	$
	we have
	$
	\sum_{i=1}^m x_i
	\geqslant (m-k+1)x_{k:m}
	\geqslant (m-k+1)(x_{k:m}-x_{j:m}),
	$
	and hence
	$
	C_g
	\leqslant {1}/({m-k+1}).
	$
	Taking $x_1=\cdots=x_{k-1}=\varepsilon$ and
	$x_k=\cdots=x_m=1$, with {$0<\varepsilon<1$}, gives
	\[
	\frac{g(x_1,\ldots,x_m)}{\sum_{i=1}^m x_i}
	=
	\frac{1-\varepsilon}
	{(k-1)\varepsilon+(m-k+1)}.
	\]
	{
		As $\varepsilon\downarrow0$,
		$
		{g(x_1,\ldots,x_m)}/{\sum_{i=1}^m x_i}
		\longrightarrow
		{1}/({m-k+1}).
		$
	}
	% which can be made arbitrarily close to $1/(m-k+1)$. 
	Thus,
	$
	C_g={1}/({m-k+1}).
	$
	
	For the linear order-statistic inequality index,
	$
	g(x_1,\ldots,x_m)=\sum_{k=1}^m a_kx_{k:m},
	$
	with $a_1\leqslant\cdots\leqslant a_m=1$ and $\sum_{k=1}^m a_k=0$.
	{ By Chebyshev's sum inequality, since both $(a_k)$ and $(x_{k:m})$ are
		nondecreasing and $\sum_{k=1}^m a_k=0$, we have
		$
		\sum_{k=1}^m a_k x_{k:m}\geqslant 0.
		$}
	Since $a_k\leqslant 1$, for all $k=1,\ldots,m$, we have
	$
	g(x_1,\ldots,x_m)\leqslant\sum_{k=1}^m x_{k:m}
	=\sum_{i=1}^m x_i,
	$
	so $C_g\leqslant 1$. Taking
	$x_1=\cdots=x_{m-1}=\varepsilon$ and $x_m=1${,
		where $0<\varepsilon<1$,} yields
	\[
	\frac{g(x_1,\ldots,x_m)}{\sum_{i=1}^m x_i}
	=
	\frac{1-\varepsilon}{1+(m-1)\varepsilon}.
	\]
	{
		As $\varepsilon\downarrow0$,
		$
		{g(x_1,\ldots,x_m)}/{\sum_{i=1}^m x_i}
		\longrightarrow 1.
		$
	}
	%which can be made arbitrarily close to $1$. 
	Therefore,
	$
	C_g=1=a_m.
	$
	
}

{
	Any function $g$ satisfying (A1) admits the representation
	\[
	g(\boldsymbol{x})
	=
	\left(\sum_{i=1}^m x_i\right)
	\phi\left(
	\frac{x_1}{\sum_{i=1}^m x_i},\ldots,
	\frac{x_m}{\sum_{i=1}^m x_i}
	\right),
	\]
	where $\phi$ is a nonnegative measurable function on
	$
	\left\{\boldsymbol{t}\in(0,1)^m:
	\sum_{i=1}^m t_i=1\right\}.
	$
	In this representation,
	\[
	C_g=\sup_{\boldsymbol{t}\in(0,1)^m: 
		\,
		\sum_{i=1}^m t_i=1}\phi(\boldsymbol{t}).
	\]
	For $m=2$, since $t_2=1-t_1$, $\phi$ can be regarded as a
	univariate function. In this case, writing $t=t_1$, homogeneity gives
	$
	C_g=\sup_{0<t<1}\phi(t,1-t)
	=\sup_{0<t<1}g(t,1-t).
	$
	%For $m\geqslant3$, $\phi$ is generally a function of
	%the full vector $\boldsymbol{t}$.
}
%\[
%T = \frac{X_1}{X_1+X_2},
%\]
%it follows that
%\begin{equation}
%	\label{eq:index_ratio}
%	I(g)
%	=
%	\frac{\mathbb{E}[\psi(T)]}{\sup_{t\in[0,1]}\psi(t)}.
%\end{equation}

This class includes several well-known indices as special cases, such as the Gini {coefficient} ($g(x_1,x_2)=|x_1-x_2|$), the smoothed CV index (SCV)  ($g(x_1,x_2)={(x_1-x_2)^2}/({x_1+x_2})$),  extended Gini-type indices, entropy-based indices, and product-type measures (see Table \ref{NPRI-index}).

% In the special case where $X \sim \mathrm{Gamma}(\alpha,\lambda)$, one has
% \[
% \frac{X_1}{X_1+\cdots+X_m} \sim \mathrm{Beta}(\alpha, (m-1)\alpha),
% \]
% so that the NPRI depends only on the distribution of this ratio.

{
	In the special case where $X \sim \mathrm{Gamma}(\alpha,\lambda)$, the
	normalized vector satisfies
	\[
	\left(
	\frac{X_1}{\sum_{i=1}^m X_i},\ldots,\frac{X_m}{\sum_{i=1}^m X_i}
	\right)^\top 
	\sim
	\mathrm{Dirichlet}(\alpha,\ldots,\alpha).
	\]
	In particular,
	\[
	\frac{X_1}{\sum_{i=1}^m X_i}\sim
	\mathrm{Beta}(\alpha,(m-1)\alpha).
	\]
	% For a general homogeneous function $g$, however, the {\color{black}NSII} depends on
	% the full vector of normalized proportions.

}

Table \ref{NPRI-index} illustrates the flexibility of the {NSII} framework, which encompasses classical inequality measures, variability indices, entropy-based quantities, and order-statistic functionals within a unified formulation.

{
	
	\section{Unbiased estimator based on a $U$-statistic of order $m$}\label{sect-3}
	
	Let $(X_1,\ldots,X_n)^\top$ be a random sample from a population with distribution
	$X\sim\mathrm{Gamma}(\alpha,\lambda)$.
	
	Table \ref{NPRI-index-1} presents particular cases of the {NSII}s listed in Table
	\ref{NPRI-index} when the population follows a {Gamma} distribution. For {Gamma}
	populations, all indices in Table \ref{NPRI-index-1} depend only on the shape
	parameter $\alpha$. This reflects the scale invariance of {NSII}s, since the rate
	parameter $\lambda$ affects only the scale of the distribution.
	\begin{table}[htb!] 
		\centering
		%\resizebox{\textwidth}{!}
		{
			\renewcommand{\arraystretch}{1.5}
			\setlength{\extrarowheight}{2.2pt}
			\begin{tabular}{|c|c|}
				\hline
				\textbf{Name} & ${I}(g)$,  $X\sim\text{Gamma}(\alpha,\lambda)$\\
				\hline
				
				Gini {coefficient}  
				& 	$ \displaystyle
				\frac{\Gamma\left(\alpha+\tfrac12\right)}
				{\sqrt{\pi}\alpha\Gamma(\alpha)}$
				\\
				\hline
				
				The $m$th Gini index 
				& $ \displaystyle		
				{1\over m\alpha\Gamma^m(\alpha)} 
				\left[
				{\displaystyle 
					\int_0^\infty 
					\{
					\Gamma^m(\alpha)-\gamma^m(\alpha,s)
					\}
					{\rm d}s
					-
					\int_0^\infty 
					\Gamma^m(\alpha,s) {\rm d}s
				}
				\right]$
				\\
				\hline
				
				%			Extended $m$th Gini index
				%			&
				%			$ \displaystyle
				%			{1\over\alpha m}
				%			\displaystyle 
				%			\sum_{r=k}^{m}
				%			(-1)^{r-k}
				%			\binom{r-1}{k-1}
				%			\binom{m}{r}
				%			\int_0^\infty 
				%			\left\{1-{\gamma^r(\alpha, t)\over\Gamma^r(\alpha) }\right\}
				%			{\rm d}t
				%			-	
				%			{1\over\alpha m}
				%			\sum_{s=j}^{m}
				%			(-1)^{s-j}
				%			\binom{s-1}{j-1}
				%			\binom{m}{s}
				%			\int_0^\infty 
				%			\left\{1-{\gamma^s(\alpha, t)\over\Gamma^s(\alpha) }\right\}
				%			{\rm d}t$
				%
				Extended $m$th Gini index
				&
				$\displaystyle
				\frac{
					{(m-k+1)}
				}{\alpha m}
				\left[
				\begin{aligned}
					&
					\sum_{r=k}^{m}
					(-1)^{r-k}
					\binom{r-1}{k-1}
					\binom{m}{r}
					\int_0^\infty
					\left\{
					1-\frac{\gamma^r(\alpha,t)}
					{\Gamma^r(\alpha)}
					\right\}
					{\rm d}t
					\\
					&-
					\sum_{s=j}^{m}
					(-1)^{s-j}
					\binom{s-1}{j-1}
					\binom{m}{s}
					\int_0^\infty
					\left\{
					1-\frac{\gamma^s(\alpha,t)}
					{\Gamma^s(\alpha)}
					\right\}
					{\rm d}t
				\end{aligned}
				\right]
				$
				\\
				\hline
				
				Linear order-statistic inequality index 
				&
				$ \displaystyle
				{1\over \alpha m} 
				\displaystyle
				\sum_{k=1}^m a_k 
				\sum_{r=k}^{m}
				\binom{m}{r}
				(-1)^{r-k}\binom{r-1}{k-1}
				{
					\displaystyle
					\int_0^\infty 
					\left\{1-{\gamma^r(\alpha, t)\over\Gamma^r(\alpha) }\right\}
					{\rm d}t
				}
				$
				\\ 
				\hline
				
				Smoothed CV index (SCV)
				& $\displaystyle \frac{1}{2\alpha+1}$ 
				\\
				\hline
				
				Power index ($p>0$) 
				& $\displaystyle
				\frac{2}{\max\{1,2^{1-p}\}}
				\,
				\frac{{\rm B}(\alpha+p,\alpha)}{{\rm B}(\alpha,\alpha)}$
				\\
				\hline
				
				Shannon entropy index 
				& $\displaystyle
				\frac{\psi(2\alpha+1)-\psi(\alpha+1)}{\log(2)}$ 
				\\
				\hline
				
				Product index 
				& $\displaystyle \frac{2\alpha}{2\alpha+1}$
				\\
				\hline
				
				{Symmetric subcase index} 
				& $ \displaystyle
				\frac{(\theta+\beta)^{\theta+\beta}}{\theta^\theta\beta^\beta}
				\,
				\frac{\Gamma(\alpha+\theta)\Gamma(\alpha+\beta)\Gamma(2\alpha)}
				{\Gamma(2\alpha+\theta+\beta)\Gamma^2(\alpha)}$, \
				{\color{black}
					$\theta=\beta$
				}
				
				\\
				\hline
				
				% Indicator index 
				% & $\displaystyle I_c(\alpha,\alpha)$
				% \\
				% \hline
				
				Maximum index  
				& 
				$\displaystyle
				{1\over 2}
				+
				\frac{\Gamma\left(\alpha+\tfrac12\right)}
				{2\sqrt{\pi}\alpha\Gamma(\alpha)}$ 
				\\
				\hline
				
				Minimum index 
				& 
				$ \displaystyle
				1
				-
				\frac{\Gamma\!\left(\alpha+\tfrac{1}{2}\right)}
				{\sqrt{\pi}\,\alpha\,\Gamma(\alpha)}$
				\\
				\hline
				
				% Linear (degenerate) 
				% & 
				% $ \displaystyle
				% \frac{1}{2}$ 
				% \\
				% \hline
				
			\end{tabular}
		}
		\caption{Examples of {NSII}s for the Gamma case. 
			In the above, $\Gamma(x)$, $\gamma(x,y)$,  $\Gamma(x,y)$,  ${\rm B}(a,b)$
			% $I_c(a,b)={\rm B}_c(a,b)/{\rm B}(a,b)$, ${\rm B}(a,b)$ 
			and $\psi(x)$ denote the (complete) {Gamma} function, the lower incomplete {Gamma} function, the upper incomplete {Gamma} function, 
			% the regularized incomplete beta function, the incomplete beta function 
			the (complete) beta function,
			and the digamma function, respectively.}
		\label{NPRI-index-1}
	\end{table}
	
	The next theorem establishes an unbiased estimator of the {NSII} $I(g)$, defined in \eqref{eq:index},
	based on a $U$-statistic of order $m$ \citep{Hoeffding1948}, where $\mu=\mathbb E[X]$ and
	%$g:(0,\infty)^m \to [0,\infty)$ satisfies Property A1.
	{ $g:(0,\infty)^m\to[0,\infty)$ is a measurable symmetric function satisfying both (A1) and (A2).}
	
	\begin{theorem}\label{thm:unbiased-estimator}
		Let $(X_1,\ldots,X_n)^\top$ be a random sample from
		$\mathrm{Gamma}(\alpha,\lambda)$, and let
		$g:(0,\infty)^m\to[0,\infty)$ be a {\color{black} measurable symmetric function
			satisfying conditions (A1)--(A2).} Then, for $n\geqslant m$, the statistic
		\begin{equation}\label{eq:estimator}
			\widehat I_{g,m}
			=
			\frac{1}{\binom{n}{m}C_{ g}}
			\,
			\frac{\displaystyle
				\sum_{1\leqslant i_1<\cdots<i_m\leqslant n}
				g(X_{i_1},\ldots,X_{i_m})
			}
			{m\overline X}
		\end{equation}
		is an unbiased estimator of the {NSII} $I(g)$.
	\end{theorem}
	
	\begin{proof}
		Set
		$
		S=\sum_{i=1}^{n}X_i,
		\
		T_i={X_i}/{S},
		\ i=1,\ldots,n.
		$
		Since $X_1,\ldots,X_n$ are independent {Gamma} random variables with common
		shape parameter $\alpha$, it follows from \citet{Mosimann1962} that
		$(T_1,\ldots,T_n)^\top$ and $S$ are independent and
		$
		(T_1,\ldots,T_n)^\top
		\sim
		\mathrm{Dirichlet}(\alpha,\ldots,\alpha).
		$
		
		By exchangeability,
		\[
		\frac{\mathbb E[g(X_1,\ldots,X_m)]}{mC_{ g}\mu}
		=
		\frac{1}{\binom{n}{m}}
		\sum_{1\leqslant i_1<\cdots<i_m\leqslant n}
		\frac{\mathbb E[g(X_{i_1},\ldots,X_{i_m})]}
		{mC_{ g}\mu}.
		\]
		
		Using the representation $X_i=ST_i$, the homogeneity of $g$, and the
		independence of $S$ and $(T_1,\ldots,T_n)^\top$, we obtain
		\begin{align*}
			\frac{\mathbb E[g(X_1,\ldots,X_m)]}{mC_{ g}\mu}
			=
			\frac{1}{\binom{n}{m}}
			\sum_{1\leqslant i_1<\cdots<i_m\leqslant n}
			\frac{\mathbb E[g(ST_{i_1},\ldots,ST_{i_m})]}
			{mC_{ g}\mu}
			=
			\frac{1}{\binom{n}{m}}
			\sum_{1\leqslant i_1<\cdots<i_m\leqslant n}
			\frac{\mathbb E[S]\,
				\mathbb E[g(T_{i_1},\ldots,T_{i_m})]}
			{mC_{ g}\mu}.
		\end{align*}
		
		Since $\mathbb E[S]=n\mu$,
		\[
		\frac{\mathbb E[g(X_1,\ldots,X_m)]}{mC_{ g}\mu}
		=
		\frac{n}{m\binom{n}{m}C_{ g}}
		\sum_{1\leqslant i_1<\cdots<i_m\leqslant n}
		\mathbb E[g(T_{i_1},\ldots,T_{i_m})].
		\]
		
		Again by homogeneity,
		$
		g(T_{i_1},\ldots,T_{i_m})
		=
		g\!\left({X_{i_1}}/{S},\ldots,{X_{i_m}}/{S}\right)
		=
		{g(X_{i_1},\ldots,X_{i_m})}/{S}.
		$
		
		Hence,
		\begin{align*}
			\frac{\mathbb E[g(X_1,\ldots,X_m)]}{mC_{ g}\mu}
			=
			\frac{n}{m\binom{n}{m}C_{ g}}
			\sum_{1\leqslant i_1<\cdots<i_m\leqslant n}
			\mathbb E\!\left[
			\frac{g(X_{i_1},\ldots,X_{i_m})}{S}
			\right]
			=
			\mathbb E\!\left[
			\frac{1}{\binom{n}{m}C_{ g}}
			\,
			\frac{\displaystyle
				\sum_{1\leqslant i_1<\cdots<i_m\leqslant n}
				g(X_{i_1},\ldots,X_{i_m})
			}
			{m\overline X}
			\right],
		\end{align*}
		because $S=n\overline X$.
		
		The expression inside the expectation is precisely the estimator 
		$\widehat I_{g,m}$
		given in \eqref{eq:estimator}. Hence, $\mathbb E[\widehat I_{g,m}]=I(g)$,
		which establishes the desired unbiasedness.
	\end{proof}
	
}

\begin{remark}
	The unbiased estimator  $\widehat{I}_{g,m}$ %\(h(X_1,\ldots,X_n)\) 
	derived in Theorem \ref{thm:unbiased-estimator} applies simultaneously to all {NSII}s introduced in Section \ref{A class of scale-invariant pairwise indices}. The specific form of the estimator depends only on the choice of the generating function \(g\).
	%
	%Moreover, the estimator $\widehat{I}_{g,m}$ %\%(h(X_1,\ldots,X_n)\) 
	%is a symmetric U-statistic of order \(m\). Consequently, its asymptotic properties may be studied using the classical theory of U-statistics.
\end{remark}

%\subsection{Asymptotic properties of $\widehat I_{g,m}$}
In the remainder of this section, we establish the large-sample properties of the unbiased estimator
$\widehat I_{g,m}$ defined in \eqref{eq:estimator}. Since
\[
\widehat I_{g,m}
=
\frac{U_n}{C_{ g}m\overline X},
\quad 
U_n=\frac{1}{\binom{n}{m}}
\,
{\displaystyle
	\sum_{1\leqslant i_1<\cdots<i_m\leqslant n}
	g(X_{i_1},\ldots,X_{i_m})
},
\]
where $U_n$ is a $U$-statistic of order $m$, its asymptotic behavior follows
from the classical theory of $U$-statistics combined with the law of large
numbers and the delta method.
\begin{proposition}[Strong consistency]
	Suppose that
	$
	\mathbb E\!\left[ |g(X_1,\ldots,X_m)| \right]<\infty.
	$
	Then
	$
	\widehat I_{g,m}
	\xrightarrow{\rm a.s.}
	I(g)
	$
	as
	$
	\ n\to\infty,
	$
	where $\xrightarrow{\rm a.s.}$ denotes almost sure convergence.
\end{proposition}
\begin{proof}
	By the strong law of large numbers for $U$-statistics 
	\citep{ Lee1990,Henze2024},
	$
	U_n
	\xrightarrow{\rm a.s.}
	\mathbb E[g(X_1,\ldots,X_m)].
	$
	Moreover,
	$
	\overline X \xrightarrow{\rm a.s.} \mu.
	$
	Hence, the proof follows from the continuous mapping Theorem.
\end{proof}

\begin{proposition}[Asymptotic normality]
	Assume that
	$
	\mathbb E[g(X_1,\ldots,X_m)^2]<\infty.
	$
	Define
	$
	\vartheta(x)
	=
	\mathbb E[g(x,X_2,\ldots,X_m)]-\mathbb E[g(X_1,\ldots,X_m)].
	$
	Then
	\[
	\sqrt{n}\{\widehat I_{g,m}-I(g)\}
	\stackrel{\mathscr{D}}{\longrightarrow}
	N(0,\nabla h^\top {\bf \Sigma} \nabla h),~~\text{as}~n\rightarrow\infty,
	\]
	where $\stackrel{\mathscr{D}}{\longrightarrow}$ denotes convergence in distribution, 
	$
	h(u,v)={u}/({C_{ g}mv}),
	$
	and
	\[
	\nabla h(\mathbb E[g(X_1,\ldots,X_m)],\mu)
	=
	\left(
	\frac{1}{C_{ g}m\mu},
	-\frac{\mathbb E[g(X_1,\ldots,X_m)]}{C_{ g}m\mu^2}
	\right)^\top,
	\ \ 
	{\bf \Sigma}
	=
	\begin{pmatrix}
		m^2\operatorname{Var}(\vartheta(X))
		&
		m\operatorname{Cov}(\vartheta(X),X)
		\\[0.2cm]
		m\operatorname{Cov}(\vartheta(X),X)
		&
		\operatorname{Var}(X)
	\end{pmatrix}.
	\]
	%and
	%\[
	%{\bf \Sigma}
	%=
	%\begin{pmatrix}
	%m^2\operatorname{Var}(\vartheta(X))
	%&
	%m\operatorname{Cov}(\vartheta(X),X)
	%\\[0.2cm]
	%m\operatorname{Cov}(\vartheta(X),X)
	%&
	%\operatorname{Var}(X)
	%\end{pmatrix}.
	%\]
\end{proposition}
\begin{proof}
	%The classical central limit theorem for $U$-statistics yields
	%\[
	%\sqrt n\,\{U_n-\mathbb E[g(X_1,\ldots,X_m)]\}
	%\stackrel{\mathscr{D}}{\longrightarrow}
	%N\!\left(0,m^2\operatorname{Var}(\vartheta(X))\right).
	%\]
	%
	Applying the classical central limit theorem of \citet[Theorem 7.3]{Hoeffding1948}, we obtain
	\[
	\sqrt n
	\begin{pmatrix}
		U_n-\mathbb E[g(X_1,\ldots,X_m)]
		\\[0,2cm]
		\overline X-\mu
	\end{pmatrix}
	\stackrel{\mathscr{D}}{\longrightarrow}
	N_2({\bf 0},{\bf \Sigma}).
	\]
	
	%Finally, 
	Since 
	$
	\widehat I_{g,m}
	=
	h(U_n,\overline X)
	$, the stated result follows from the multivariate delta method. This completes the proof.
\end{proof}

\section{Simulation study}\label{Simulation study}

% In this section, Monte Carlo simulations are presented to assess the estimator $\widehat{I}_{g,m}$ of the NPRI $I(g)$.
% We analyze replications of the random sample $X_1^{(i)}, \ldots, X_n^{(i)}$ $(i=1,\ldots,M)$ generated by a $\operatorname{Gamma}(\alpha, \lambda)$ distribution.
% We assess the performance of the estimator using the Bias 
% $$ \operatorname{Bias} = \frac{1}{M} \sum_{i=1}^M \left[\widehat{I}_{g,m}^{(i)} - I(g,\alpha) \right] $$ 
% and the root mean square error (RMSE)
%  $$RMSE = \sqrt{\frac{1}{M} \sum_{i=1}^M \left[\widehat{I}_{g,m}^{(i)} - I(g,\alpha) \right]^2}$$
% under different parameter scenarios. Here
%  $\widehat{I}_{g,m}^{(i)}=	h(X_1^{(i)},\ldots,X_n^{(i)})$ is given in \eqref{eq:estimator}.
% We also examine the sensitivity of the estimator to different values of $m$, and assess its stabilization in more general populations.

\subsection{Evaluating the Bias and RMSE of estimates}

In this subsection, we present a Monte Carlo simulation study to evaluate the estimator $\widehat{I}_{g,m}$ \eqref{eq:estimator} for $I(g)$ \eqref{eq:index}. 
We conducted $M=5,000$ Monte Carlo replications, generating samples $X_1, X_2,\ldots, X_n$ from the $\operatorname{Gamma}(\alpha, \lambda)$ considering for the simulations the true population's parameter $\lambda=1$ and $\alpha \in\{1,3,5\}$, and the {sample sizes} $n\in\{10, 20, 50\}$.
Moreover, we initially chose the parameter $m=2$. %\sout{A sensitivity analysis of the results with respect to the choice of $m$ is presented} 
{\color{black}Other kernel orders and configurations are examined in} Subsection \ref{sec:sim_chosing_m} and in real data applications (Section \ref{sec:applications}).

Table \ref{tab:mc_alpha} presents the Bias and RMSE
of the estimates for the following {NSII}s: Gini, Power ($p=0.5$), {Symmetric subcase} ($\theta=\beta=1$), Maximum, Minimum, SCV, and Product. %\sout{Observe that as the sample size $n$ increases,  both the Bias and the RMSE of the Monte Carlo simulations generally decrease.} 
{\color{black}
	Since $\widehat{I}_{g,m}$ is unbiased under the Gamma distribution (Theorem
	\ref{thm:unbiased-estimator}), the empirical biases reported in Table \ref{tab:mc_alpha} are only Monte Carlo error and need not decrease with $n$. On the other hand, the RMSE decreases with $n$ in every row.
}

\begin{table}[htb!]
	\centering
	\caption{Bias and RMSE for the estimates of {NSII} based on Monte Carlo simulations from $X \sim \mathrm{Gamma}(\alpha,\lambda=1)$, with $m = 2$, $n\in\{10,20,50\}$, $\alpha\in\{1,3,5\}$ and $M = 5,000$ replications.}
	\label{tab:mc_alpha}
	\begin{tabular}{lllcccccc}
		\hline
		&& &
		\multicolumn{2}{c}{$n=10$} &
		\multicolumn{2}{c}{$n=20$} &
		\multicolumn{2}{c}{$n=50$} \\
		\cline{4-5}\cline{6-7}\cline{8-9}
		$\alpha$ & Index & $I(g)$ &
		Bias & RMSE &
		Bias & RMSE &
		Bias & RMSE \\
		\hline
		1 &   Gini                          & $0.50000$ & ${-0.00134}$ & $0.09490$ & ${-0.00043}$ & $0.06645$ & ${0.00118}$ & $0.04083$ \\
		1 &    Power               & $0.94281$ & ${0.00001}$ & $0.02158$ & ${-0.00002}$ & $0.01507$ & ${-0.00027}$ & $0.00933$ \\
		1 &    {Symmetric subcase} & $0.66667$ & ${0.00339}$ & $0.10406$ & ${-0.00033}$ & $0.07074$ & ${-0.00040}$ & $0.04543$ \\
		1 &    Maximum                       & $0.75000$ & ${-0.00127}$ & $0.04877$ & ${0.00046}$ & $0.03255$ & ${0.00007}$ & $0.02056$ \\
		1 &    Minimum                       & $0.50000$ & ${-0.00123}$ & $0.09699$ & ${0.00038}$ & $0.06722$ & ${-0.00140}$ & $0.04083$ \\
		1 &   SCV                           & $0.33333$ & ${-0.00017}$ & $0.10228$ & ${-0.00060}$ & $0.07182$ & ${0.00050}$ & $0.04535$ \\
		1 &  Product                       & $0.66667$ & ${0.00112}$ & $0.10302$ & ${0.00036}$ & $0.07223$ & ${-0.00018}$ & $0.04494$ \\
		\hline
		3&    Gini                          & $0.31250$ & ${-0.00096}$ & $0.07044$ & ${0.00097}$ & $0.04836$ & ${-0.00032}$ & $0.02900$ \\
		3&     Power                & $0.97954$ & ${0.00000}$ & $0.00893$ & ${-0.00012}$ & $0.00630$ & ${0.00004}$ & $0.00384$ \\
		3&     {Symmetric subcase}  & $0.85714$ & ${-0.00129}$ & $0.05676$ & ${0.00011}$ & $0.03831$ & ${-0.00011}$ & $0.02392$ \\
		3&      Maximum                       & $0.65625$ & ${-0.00049}$ & $0.03412$ & ${-0.00006}$ & $0.02400$ & ${-0.00039}$ & $0.01498$ \\
		3&     Minimum                       & $0.68750$ & ${0.00109}$ & $0.06883$ & ${0.00115}$ & $0.04772$ & ${-0.00003}$ & $0.02971$ \\
		3&      SCV                           & $0.14286$ & ${0.00120}$ & $0.05708$ & ${0.00057}$ & $0.03816$ & ${-0.00020}$ & $0.02442$ \\
		3&      Product                       & $0.85714$ & ${-0.00003}$ & $0.05646$ & ${-0.00027}$ & $0.03854$ & ${-0.00026}$ & $0.02453$ \\
		\hline
		5&  Gini                          & $0.24609$ & ${-0.00084}$ & $0.05528$ & ${0.00002}$ & $0.03886$ & ${-0.00006}$ & $0.02406$ \\
		5&     Power   & $0.98761$ & ${0.00005}$ & $0.00558$ & ${-0.00006}$ & $0.00392$ & ${0.00003}$ & $0.00242$ \\
		5&     {Symmetric subcase}  & $0.90909$ & ${0.00071}$ & $0.03734$ & ${-0.00051}$ & $0.02653$ & ${-0.00024}$ & $0.01646$ \\
		5&     Maximum                       & $0.62305$ & ${-0.00007}$ & $0.02841$ & ${0.00021}$ & $0.01927$ & ${0.00012}$ & $0.01208$ \\
		5&     Minimum                       & $0.75391$ & ${-0.00014}$ & $0.05609$ & ${-0.00037}$ & $0.03820$ & ${0.00006}$ & $0.02381$ \\
		5&     SCV                           & $0.09091$ & ${-0.00023}$ & $0.03840$ & ${0.00040}$ & $0.02642$ & ${0.00016}$ & $0.01666$ \\
		5&    Product                       & $0.90909$ & ${-0.00074}$ & $0.03742$ & ${-0.00010}$ & $0.02621$ & ${-0.00007}$ & $0.01623$ \\
		\hline
	\end{tabular}
\end{table}

% \subsection{Estimator sensitivity}\label{sec:sim_chosing_m}
\subsection{{\color{black}Higher-order $m$ and weight configurations}}\label{sec:sim_chosing_m}

Recall that the estimator $\widehat{I}_{g,m}$ \eqref{eq:estimator} depends on the $m$ value. 
%In this subsection, \sout{we analyze the estimator's sensitivity to increasing the $m$ value.}
{\color{black}
	Each order $m$ and each configuration $(j,k)^\top$ or $(a_1,\ldots,a_m)^\top$ specifies a different
	generating function $g$, hence a possible different $I(g)$. }
Table \ref{tab:mc_m_alpha} presents the Bias and RMSE for the estimates of $m$th Gini, Extended $m$th Gini, and Linear Order-Statistic Inequality indices, {\color{black}under several configurations},
based on {$M = 5,000$} Monte Carlo simulations from $X \sim \mathrm{Gamma}(\alpha,\lambda=1)$, $\alpha\in\{1,3,5\}$.  
%For the Extended $m$th Gini, the results indicate that, under the same parameter configuration $(j,k)$, the Bias is generally smaller. 
{\color{black} Therefore, Table \ref{tab:mc_m_alpha} compares the performance of the estimator $\widehat{I}_{g,m}$ \eqref{eq:estimator} across several target functionals, rather than the sensitivity of a
	single estimator to a tuning parameter. As in Table \ref{tab:mc_alpha}, the biases are of the order of
	the Monte Carlo error and the RMSE decreases with $n$.}

%{\color{black}
	%Depending on the choice of $m$ and the configurations, the Linear Order-Statistic yields the same values. For example, in Table \ref{tab:mc_m_alpha}, with $m=4$ and configuration $(-1, -1/3 , 1/3 , 1)^\top$ would generate the same results as $m=3$ with $(-1,0,1)^\top$}

{
	\begin{remark}
		Note that without a careful choice of configurations, the results reported in Table \ref{tab:mc_m_alpha} could be identical across different settings. For example, the Linear Order-Statistic with $m=4$ and $\left(-1,-{1}/{3},{1}/{3},1\right)^\top$ would be identically equal, for every sample, to the $m=3$ estimator based on $(-1,0,1)^\top.$ %On the other hand, the adopted vector $(-1,-1,1,1)^\top$ which defines a genuinely different functional. 
		{The adopted vector $(-1,-1,1,1)^\top$ defines a genuinely different functional
			rather than a rescaling of the former.}
		Table \ref{tab:mc_m_alpha} therefore presents results for configurations that avoid redundancy and represent distinct cases.
	\end{remark}
}

\begin{table}[htb!]
	\centering
	\caption{Bias and RMSE for the estimates of {NSII} based on Monte Carlo simulations from $X \sim \mathrm{Gamma}(\alpha,\lambda=1)$, with $n\in\{10,20,50\}$, $\alpha\in\{1,3,5\}$ and {$M = 5,000$} replications.}
	\label{tab:mc_m_alpha}
	\setlength{\tabcolsep}{0.5pt}
	{
		\begin{tabular}{ccccccccccc}
			\hline
			&&&& &
			\multicolumn{2}{c}{$n=10$} &
			\multicolumn{2}{c}{$n=20$} &
			\multicolumn{2}{c}{$n=50$} \\
			\cline{6-7}\cline{8-9}\cline{10-11}
			$\alpha$ & Index & Configurations & $m$ & $I(g)$ &
			Bias & RMSE &
			Bias & RMSE &
			Bias & RMSE \\
			\hline
			1 & $m$th Gini & -- & $3$ & $0.50000$ & $-0.00060$ & $0.09551$ & $-0.00018$ & $0.06725$ & $-0.00086$ & $0.04183$ \\
			&  & -- & $4$ & $0.45833$ & $-0.00054$ & $0.09102$ & $-0.00019$ & $0.06414$ & $-0.00080$ & $0.03974$ \\
			& Extended $m$th Gini & $(1, 2)^\top$ & $3$ & $0.33333$ & $-0.00015$ & $0.06135$ & $-0.00004$ & $0.03747$ & $-0.00044$ & $0.02192$ \\
			&  & $(1, 2)^\top$ & $4$ & $0.25000$ & $0.00013$ & $0.06253$ & $-0.00009$ & $0.03875$ & $-0.00026$ & $0.02248$ \\
			&  & $(2, 3)^\top$ & $3$ & $0.33333$ & $-0.00052$ & $0.10003$ & $-0.00017$ & $0.06996$ & $-0.00064$ & $0.04295$ \\
			&  & $(2, 3)^\top$ & $4$ & $0.25000$ & $-0.00035$ & $0.06335$ & $0.00004$ & $0.03821$ & $-0.00041$ & $0.02288$ \\
			& Linear Order-Statistic & %$(-1,1)^\top$ & $2$ & $0.50000$ & $-0.00060$ & $0.09551$ & $-0.00018$ & $0.06725$ & $-0.00086$ & $0.04183$ \\
			$(-1,0,1)^\top$ & $3$ & $0.50000$ & $-0.00060$ & $0.09551$ & $-0.00018$ & $0.06725$ & $-0.00086$ & $0.04183$ \\
			&  & $(-\tfrac{2}{3},-\tfrac{1}{3},1)^\top$ & $3$ & $0.44444$ & $-0.00057$ & $0.09595$ & $-0.00018$ & $0.06759$ & $-0.00079$ & $0.04189$ \\
			% &  & {\color{black}$(-1,-\tfrac{1}{3},\tfrac{1}{3},1)^\top$} & {\color{black}$4$} & {\color{black}$0.50000$} & {\color{black}$-0.00060$} & {\color{black}$0.09551$} & {\color{black}$-0.00018$} & {\color{black}$0.06725$} & {\color{black}$-0.00086$} & {\color{black}$0.04183$}\\
			&  & $(-1,-1,1,1)^\top$ & $4$ & $0.58333$ & $-0.00071$ & $0.10706$ & $-0.00017$ & $0.07471$ & $-0.00100$ & $0.04666$ \\
			\hline
			3 & $m$th Gini & -- & $3$ & $0.31250$ & $-0.00107$ & $0.06895$ & $-0.00112$ & $0.04672$ & $-0.00038$ & $0.02930$ \\
			&  & -- & $4$ & $0.28545$ & $-0.00111$ & $0.06358$ & $-0.00099$ & $0.04305$ & $-0.00036$ & $0.02695$ \\
			& Extended $m$th Gini & $(1, 2)^\top$ & $3$ & $0.25291$ & $0.00038$ & $0.05593$ & $-0.00124$ & $0.03598$ & $-0.00012$ & $0.02179$ \\
			&  & $(1, 2)^\top$ & $4$ & $0.21709$ & $0.00037$ & $0.05835$ & $-0.00106$ & $0.03607$ & $-0.00008$ & $0.02137$ \\
			&  & $(2, 3)^\top$ & $3$ & $0.18604$ & $-0.00127$ & $0.05844$ & $-0.00050$ & $0.03936$ & $-0.00032$ & $0.02435$ \\
			&  & $(2, 3)^\top$ & $4$ & $0.16228$ & $0.00021$ & $0.04538$ & $-0.00080$ & $0.02901$ & $-0.00010$ & $0.01753$ \\
			& Linear Order-Statistic & %$(-1,1)^\top$ & $2$ & $0.31250$ & $-0.00107$ & $0.06895$ & $-0.00112$ & $0.04672$ & $-0.00038$ & $0.02930$ \\ &  &
			$(-1,0,1)^\top$ & $3$ & $0.31250$ & $-0.00107$ & $0.06895$ & $-0.00112$ & $0.04672$ & $-0.00038$ & $0.02930$ \\
			&  & $(-\tfrac{2}{3},-\tfrac{1}{3},1)^\top$ & $3$ & $0.27035$ & $-0.00114$ & $0.06430$ & $-0.00091$ & $0.04359$ & $-0.00036$ & $0.02727$ \\
			% &  & {\color{black}$(-1,-\tfrac{1}{3},\tfrac{1}{3},1)^\top$} & {\color{black}$4$} & {\color{black}$0.31250$} & {\color{black}$-0.00107$} & {\color{black}$0.06895$} & {\color{black}$-0.00112$} & {\color{black}$0.04672$} & {\color{black}$-0.00038$} & {\color{black}$0.02930$} \\
			&  & $(-1,-1,1,1)^\top$ & $4$ & $0.36659$ & $-0.00100$ & $0.08076$ & $-0.00138$ & $0.05462$ & $-0.00041$ & $0.03425$ \\
			\hline
			5 & $m$th Gini & -- & $3$ & $0.24609$ & $-0.00105$ & $0.05662$ & $-0.00017$ & $0.03850$ & $0.00019$ & $0.02391$ \\
			&  & -- & $4$ & $0.22467$ & $-0.00092$ & $0.05191$ & $-0.00011$ & $0.03529$ & $0.00018$ & $0.02187$ \\
			& Extended $m$th Gini & $(1, 2)^\top$ & $3$ & $0.20992$ & $-0.00148$ & $0.04940$ & $-0.00073$ & $0.03247$ & $0.00033$ & $0.01992$ \\
			&  & $(1, 2)^\top$ & $4$ & $0.18633$ & $-0.00146$ & $0.05259$ & $-0.00074$ & $0.03346$ & $0.00046$ & $0.01981$ \\
			&  & $(2, 3)^\top$ & $3$ & $0.14113$ & $-0.00031$ & $0.04503$ & $0.00020$ & $0.03013$ & $0.00002$ & $0.01855$ \\
			&  & $(2, 3)^\top$ & $4$ & $0.12855$ & $-0.00075$ & $0.03714$ & $-0.00035$ & $0.02370$ & $0.00003$ & $0.01468$ \\
			& Linear Order-Statistic & %$(-1,1)^\top$ & $2$ & $0.24609$ & $-0.00105$ & $0.05662$ & $-0.00017$ & $0.03850$ & $0.00019$ & $0.02391$ \\    &  &
			$(-1,0,1)^\top$ & $3$ & $0.24609$ & $-0.00105$ & $0.05662$ & $-0.00017$ & $0.03850$ & $0.00019$ & $0.02391$ \\
			&  & $(-\tfrac{2}{3},-\tfrac{1}{3},1)^\top$ & $3$ & $0.21111$ & $-0.00080$ & $0.05175$ & $-0.00005$ & $0.03510$ & $0.00013$ & $0.02177$ \\
			%  &  & {\color{black}$(-1,-\tfrac{1}{3},\tfrac{1}{3},1)^\top$ }& {\color{black}$4$} & {\color{black}$0.24609$} & {\color{black}$-0.00105$} & {\color{black}$0.05662$} & {\color{black}$-0.00017$} & {\color{black}$0.03850$} & {\color{black}$0.00019$} & {\color{black}$0.02391$} \\
			&  & $(-1,-1,1,1)^\top$ & $4$ & $0.28894$ & $-0.00130$ & $0.06679$ & $-0.00029$ & $0.04527$ & $0.00020$ & $0.02818$ \\
			\hline
		\end{tabular}
	}
\end{table}

\section{Applications}\label{sec:applications}
% In order to show the versatility of the {NSII} estimator, we show that the framework successfully fits the distribution of Gross Domestic Product (GDP) per capita across the Americas. The dataset includes $n = 34$ countries and territories in the Americas in 2023, measured in international dollars at 2021 PPP prices. The data were obtained from \url{https://ourworldindata.org/grapher/gdp-per-capita-worldbank} (accessed on June 2026) and converted to USD$\times10^3$. The sample spans the full income spectrum of the region, from low-income economies to high-income countries. 

In order to illustrate the versatility of the {NSII} estimator, {we apply the framework to
the distribution of} Gross Domestic Product (GDP) per capita across the Americas. %% >>> CONFIRMAR ANTES DE SUBMETER: (i) rotulo/codigo exato da variavel; (ii) data de extracao; (iii) regra de exclusao <<<
{ To ensure replicability of our results and provide support for future users of inequality indices, all R code is publicly available at \url{https://github.com/FSQuintino/inequality-indices} (accessed on 30 September 2026).
The variable is GDP per capita, measured in international dollars at 2021 PPP prices, and it was obtained from the Our World in Data repository
(\url{https://ourworldindata.org/grapher/gdp-per-capita-worldbank}, accessed in June 2026), with
reference year 2023, expressed in thousands of 2021 international dollars, which are not nominal US
dollars. The sample consists of the $n=34$ economies of the Americas. Countries without a 2023 value were excluded, and no imputation was performed. For the reader’s convenience, the dataset can be accessed in R using the following command:
\begin{lstlisting}
	countries <- c("Haiti", "Honduras", "Nicaragua", "Bolivia", "Jamaica", "El Salvador",
	"Guatemala", "Belize", "Ecuador", "Peru", "Paraguay", "Grenada", "Dominica",
	"St. Vincent and the Grenadines", "Colombia", "Brazil", "Suriname", "Barbados", "Mexico",
	"Dominican Republic", "Saint Lucia", "Costa Rica", "Argentina", "Antigua and Barbuda",
	"Chile", "St. Kitts and Nevis", "Uruguay", "Trinidad and Tobago", "Bahamas", "Panama",
	"Puerto Rico", "Guyana", "Canada", "United States")
	
	gdp <- c(2.956, 6.468, 7.487, 9.844, 10.291, 11.404, 12.389, 12.455, 14.472, 15.294,
	15.783, 16.946, 17.420, 18.335, 18.692, 19.018, 19.044, 19.224, 22.143, 23.088, 23.403,
	26.293, 27.105, 29.463, 30.409, 31.019, 31.706, 33.106, 35.864, 42.995, 49.315,
	55.919, 74.578)
	
	gdp_data <- data.frame(country = countries, gdp = gdp)
\end{lstlisting}
}

%Figure \ref{fig:mapa} shows the spatial distribution of the data.
% \begin{figure}[H]
%     \centering
%     \includegraphics[width=1.0\linewidth, height=0.5\textheight]{figs/mapa.pdf}
%     \caption{Spatial distribution of income levels in the Americas.}
%     \label{fig:mapa}
% \end{figure}

\begin{figure}[htb!]
\centering
\includegraphics[width=0.9\linewidth, height=0.3\textheight]{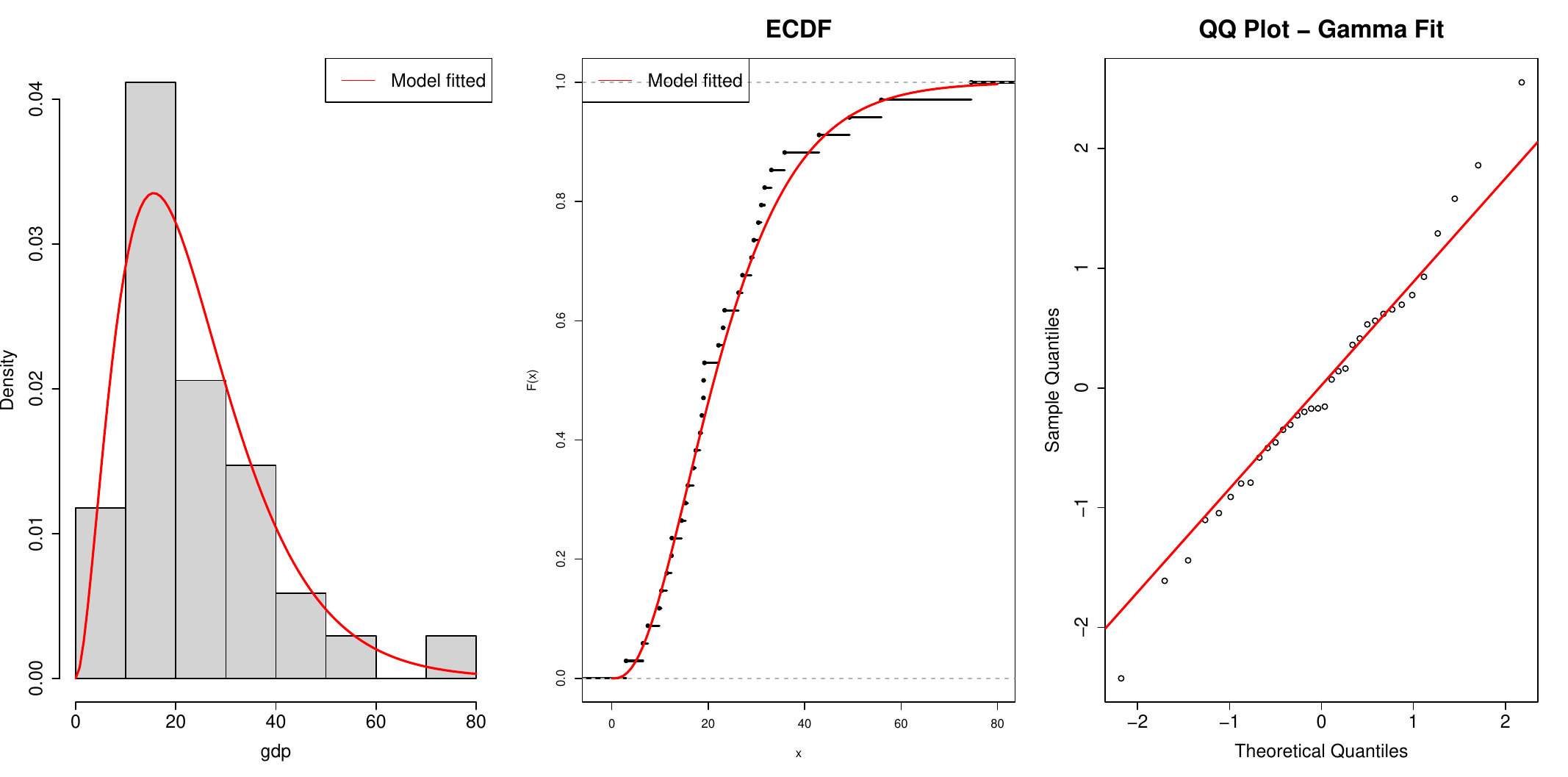}
\caption{Histograms with fitted Gamma density (left), empirical CDF with fitted Gamma distribution (middle), and Normal QQ plot of randomized residuals from the fitted Gamma model (right).}
\label{fig:gamma_fit}
\end{figure}

% Figure \ref{fig:gamma_fit} shows that the distribution of GDP data is right-skewed. This supports the initial choice of using the Gamma model.
% The Gamma model with parameters $\widehat{\alpha}=2.876$ and $\widehat{\lambda}=0.120$ appears to be a good choice, as indicated by the histogram, the empirical CDF (ECDF), and the randomized residuals.
% Parametric bootstrap Kolmogorov–Smirnov ($D = 0.091$, $\operatorname{p-value} = 0.91$), Cramér–von Mises ($\omega^2=0.16924,~ \operatorname{p-value} = 0.9191$), and Anderson-Darling ($A^2=0.5456,~ \operatorname{p-value} = 0.66$) goodness-of-fit tests ensure statistical validity for the Gamma model with parameters estimated from the data.

Figure \ref{fig:gamma_fit} shows that the distribution of GDP data is right-skewed{, which
motivates us to consider a Gamma model}. The Gamma model with parameters $\widehat{\alpha}=2.876$ and
$\widehat{\lambda}=0.120$ {estimated by maximum likelihood (ML) has the PDF and CDF compatible with the histogram and the empirical CDF (ECDF). 
The Kolmogorov--Smirnov ($D=0.0913$, $\operatorname{p-value} =0.68$), Cram\'er--von Mises ($\omega^2=0.0294$, $\operatorname{p-value} =0.87$) and Anderson--Darling ($A^2=0.1936$, $\operatorname{p-value} =0.91$) statistics were calibrated by a parametric bootstrap with $B=10,000$ samples of size $n$ from $\mathrm{Gamma}(\widehat{\alpha},\widehat{\lambda})$, in which both parameters are re-estimated by ML. The tests therefore provide no statistically significant evidence against the fitted Gamma model.}

Based on the GDP data, Table \ref{tab:indices} presents the
estimates \eqref{eq:estimator} for the following {NSII}s: Gini, Power, {Symmetric subcase}, Maximum, Minimum, SCV, and Product. 
%\sout{The results for the GDP data indicate a moderate to high level of economic inequality in the Americas, with each of the {NSII}s providing a distinct interpretation.}
{\color{black}
The indices do not share a common orientation. Actually, under perfect equality $I(g)=0$ for the Gini and SCV
indices, $1/2$ for the Maximum index and $1$ for the Power, Minimum, Product and {Symmetric subcase} indices.
%Moreover, the first two increase with inequality while the others decrease. %%In the original scale, the Gini estimate $0.3286$ means that two countries drawn at random differ, on average, by $66\%$ of the regional mean ($15.7$ against $\widehat{\mu}=23.9$ thousand international dollars).
{Moreover, the Gini, SCV and Maximum indices increase with inequality, whereas the Power,
	Minimum, Product and Symmetric subcase indices decrease. Indices with opposite orientations should not
	be compared through their numerical magnitudes.}
Furthermore, some of these indices are
deterministically related, both for $I(g)$ and for the estimator, as we point out below:
\begin{itemize}
	\item[(i)] the {Symmetric subcase index} with
	$\theta=\beta=1$ is the Product index;
	\item[(ii)] $$\widehat{I}_{g,2}^{\text{Max}}=\frac{1}{2}(1+\widehat{I}_{g,2}^\text{Gini}), ~~\widehat{I}_{g,2}^\text{Min}=1-\widehat{I}_{g,2}^\text{Gini}, ~~\text{and}~~\widehat{I}_{g,2}^\text{SCV}=1-\widehat{I}_{g,2}^\text{Prod}.$$
\end{itemize}
Thus, Table \ref{tab:indices} presents three distinct indices (Gini, Product, and Power). The remaining entries serve as a consistency check.
}

% The 95\% {\color{red}b}ootstrap confidence intervals (CIs) with $B=5{,}000$ replications are also reported. 
% {\color{black}
%     The standard errors (SEs) correspond to the standard deviations of the bootstrap distributions. In Table \ref{tab:mGini}, the non-parametric columns resample the $n=34$ observations with replacement. The parametric columns draw each resample $\mathbf{X}^{(b)} = (X_1^{(b)},X_2^{(b)}, \cdots, X_n^{(b)})^\top$ ($b=1,\cdots,5000$) from $\mathrm{Gamma}(\widehat{\alpha},\widehat{\lambda})$, re-estimate the shape parameter by maximum {\color{red}likelihood and} compute
% $$\widehat{\alpha}^{(b)} = \widehat{\alpha}^{(b)}(\mathbf{X}^{(b)}).$$
% Finally, the closed-form expressions of 
% ${I}_{m}(\widehat{\alpha}^{(b)})$ are evaluated, and the SEs are obtained as
% $$\operatorname{SE} = \operatorname{sd}({I}_{m}(\widehat{\alpha}^{(1)}), \cdots, {I}_{m}(\widehat{\alpha}^{(5,000)})).$$
% }

{
All 95\% bootstrap confidence intervals (CIs) in Tables \ref{tab:indices} and \ref{tab:mGini} are percentile intervals
$\big[\widehat{I}^{*(0.025)}_{g,m},\,\widehat{I}^{*(0.975)}_{g,m}\big]$, the empirical $2.5\%$ and $97.5\%$
quantiles of the $B=5{,}000$ replicates, and the standard errors (SEs) correspond to the standard deviations of the bootstrap replications. In Table \ref{tab:indices} and
in the non-parametric CIs of Table \ref{tab:mGini} the resampling unit is the economy, that is, the $n=34$
observations are drawn with replacement, and the estimator \eqref{eq:estimator} is recomputed in each resample. On the other hand, the
parametric CIs of Table \ref{tab:mGini} draw each resample
$\mathbf{X}^{(b)}=(X_1^{(b)},\ldots,X_n^{(b)})^\top$ $(b=1,\ldots,B)$, from
$\mathrm{Gamma}(\widehat{\alpha},\widehat{\lambda})$ and re-estimate both parameters by ML,
$\widehat{\lambda}^{(b)}=\widehat{\alpha}^{(b)}/\overline{X}^{(b)}$ with $\widehat{\alpha}^{(b)}$ solving the profile
score equation, evaluating ${I}_{m}(\widehat{\alpha}^{(b)})$ of Table \ref{NPRI-index-1}. %The same resamples are used in both tables.}
The same set of non-parametric resamples is used in Table \ref{tab:indices} and in the
non-parametric columns of Table \ref{tab:mGini}.

% \begin{table}[htb!]
%   \centering
%   \caption{Estimates of the {\color{black}NSII} for the GDP data and 95\% {\color{red}b}ootstrap confidence intervals, with $m=2$ and $B=5,000$ resamples.}
%   \label{tab:indices}
%   \begin{tabular}{ccc}
	%     \hline
	%      Index & $\widehat{I}_{g,2}$ & 95\%CI \\
	%     \hline
	%     Gini                                       & $0.3286$ & $[0.2368, 0.3966]$  \\
	%     Power ($p=0.5$)                            & $0.9777$ & $[0.9673, 0.9880]$  \\
	%     Asymmetric ($\theta=\beta=1$)              & $0.8456$ & $[0.7838, 0.9124]$  \\
	%     Maximum                                    & $0.6643$ & $[0.6201, 0.6979]$  \\
	%     Minimum                                    & $0.6714$ & $[0.6026, 0.7596]$  \\
	%     SCV                                        & $0.1544$ & $[0.0880, 0.2146]$ \\
	%     Product                                    & $0.8456$ & $[0.7810, 0.9113]$  \\
	%     %Indicator ($c=0.5$)                        & $1.00000$ & $0.38466$ & $0.65015$  \\
	%     \hline
	%   \end{tabular}
% \end{table}

\begin{table}[htb!]
\centering
\caption{Estimates of the {NSII} for the GDP data and 95\% {b}ootstrap confidence intervals, with $m=2$ and $B=5,000$ resamples. {Bootstrap standard errors (SEs) are given in parentheses.}}
\label{tab:indices}
{
	\begin{tabular}{ccc}
		\hline
		Index & $\widehat{I}_{g,2}$ (SE) & 95\% CI \\
		\hline
		Gini                                       & $0.3286\,(0.0397)$ & $[0.2407, 0.3951]$  \\
		Power ($p=0.5$)                            & $0.9777\,(0.0053)$ & $[0.9674, 0.9878]$  \\
		{Symmetric subcase} ($\theta=\beta=1$)              & $0.8456\,(0.0327)$ & $[0.7835, 0.9107]$  \\
		Maximum                                    & $0.6643\,(0.0199)$ & $[0.6203, 0.6975]$  \\
		Minimum                                    & $0.6714\,(0.0397)$ & $[0.6049, 0.7593]$  \\
		SCV                                        & $0.1544\,(0.0327)$ & $[0.0893, 0.2165]$  \\
		Product                                    & $0.8456\,(0.0327)$ & $[0.7835, 0.9107]$  \\
		\hline
	\end{tabular}
}
\end{table}

Once the Gamma distribution is fitted to the GDP data, we can compare the estimated values of  $I(g)$ obtained from the non-parametric estimator $\widehat{I}_{g,m}$ \eqref{eq:estimator} with those from %the parametric form $I_m(\widehat{\alpha}, \widehat{\lambda})$ (Table \ref{NPRI-index-1}). In addition, the sensitivity to different values of $m$ can be studied. 
{$I_m(\widehat{\alpha})$ (Table \ref{NPRI-index-1}), which depends on the
shape parameter only, by scale invariance. Different target functionals, that is, indices generated by different orders $m$ and configurations, can also be compared.} Table \ref{tab:mGini} presents the estimates for the $m$th Gini, Extended $m$th Gini, and Linear Order-Statistic Inequality indices. Overall, the parametric 95\% CI tends to be narrower.

\begin{table}[htb!]
\centering
\caption{Parametric ${I}_{m}(\widehat{\alpha})$ and non-parametric $\widehat{I}_{g,m}$ estimates of $m$th Gini, Extended $m$th Gini, and Linear Order-Statistic Inequality indices for the GDP data, with 95\% {b}ootstrap confidence intervals ($B=5,000$ resamples), varying $m$, $(j, k)^\top$ and $(a_1,\ldots,a_m)^\top$ configurations. {Bootstrap SEs are given in parentheses.}}
\label{tab:mGini}
	\setlength{\tabcolsep}{0.0pt}
{
	\begin{tabular}{ccccccc}
		\hline
		& & &
		\multicolumn{2}{c}{Non-parametric} &
		\multicolumn{2}{c}{Parametric} \\
		\cline{4-5}\cline{6-7}
		Index & $m$ & Configuration &
		$\widehat{I}_{g,m}$ (SE) & 95\% CI &
		$I_m(\widehat{\alpha})$ (SE) & 95\% CI \\
		\hline
		$m$th Gini & $3$ & -- & $0.3286\,(0.0397)$ & $[0.2407, 0.3951]$ & $0.3186\,(0.0336)$ & $[0.2451, 0.3776]$ \\
		& $4$ & -- & $0.3029\,(0.0367)$ & $[0.2200, 0.3624]$ & $0.2910\,(0.0308)$ & $[0.2238, 0.3452]$ \\
		\hline
		Extended $m$th Gini & $3$ & $(1, 2)^\top$ & $0.2443\,(0.0286)$ & $[0.1845, 0.2974]$ & $0.2565\,(0.0201)$ & $[0.2092, 0.2883]$ \\
		& $3$ & $(2, 3)^\top$ & $0.2064\,(0.0334)$ & $[0.1320, 0.2604]$ & $0.1903\,(0.0236)$ & $[0.1405, 0.2335]$ \\
		& $4$ & $(1, 2)^\top$ & $0.2125\,(0.0291)$ & $[0.1510, 0.2654]$ & $0.2195\,(0.0133)$ & $[0.1858, 0.2381]$ \\
		& $4$ & $(2, 3)^\top$ & $0.1540\,(0.0222)$ & $[0.1132, 0.2009]$ & $0.1653\,(0.0168)$ & $[0.1280, 0.1944]$ \\
		& $4$ & $(3, 4)^\top$ & $0.1551\,(0.0292)$ & $[0.0880, 0.2006]$ & $0.1352\,(0.0180)$ & $[0.0978, 0.1687]$ \\
		\hline
		Linear Order-Statistic & $4$ & $(-1,0,0,1)^\top$ & $0.3029\,(0.0367)$ & $[0.2200, 0.3624]$ & $0.2910\,(0.0308)$ & $[0.2238, 0.3452]$ \\
		& $5$ & $(-1,0,0,0,1)^\top$ & $0.2773\,(0.0339)$ & $[0.1996, 0.3310]$ & $0.2635\,(0.0280)$ & $[0.2024, 0.3128]$ \\
		\hline
		%  Configuracoes alternativas (fator m-k+1 = 1, valores inalterados):
		%   &  $3$ & $(1,3)^\top$ &  $0.3286$ & $[0.2396, 0.3957]$  &  $0.3186$ & $[0.2457, 0.3766]$  \\
		%   &  $4$ & $(1,4)^\top$ &  $0.3029$ & $[0.2203, 0.3634]$  &  $0.2910$ & $[0.2243, 0.3443]$  \\
	\end{tabular}
}
\end{table}

\section{Concluding remarks}\label{sect-5}

For {Gamma} populations, explicit expressions were obtained for several members of the class. By exploiting the independence between the total sum and the associated Dirichlet proportions, we derived a simple U-statistic-based estimator and proved its unbiasedness for any NSII {in the proposed class}. Thus, a single estimation procedure simultaneously covers a broad family of normalized scale-invariant indices.
This paper introduced a general class of {NSIIs} generated by 
{ measurable symmetric}
homogeneous functions.
%homogeneous functions. 
The proposed framework unifies a variety of well-known measures, including the Gini coefficient, generalized Gini indices, entropy-based indices, variability measures, and order-statistic functionals.
% Monte Carlo simulations attested to the performance of the analytical closed-form expressions hereby derived. By applying our methodology to real-world GDP data, we can conduct a deeper examination of economic inequalities in the Americas.
{Monte Carlo simulations were conducted to assess the finite-sample performance of the proposed unbiased estimator.}
% Future work may investigate asymptotic properties of the proposed estimator, variance estimation, confidence intervals, and extensions to distributional families beyond the {\color{red}Gamma} model.
{Future work may focus on consistent estimation of the asymptotic
variance, the construction and theoretical justification of confidence
intervals, efficiency comparisons with alternative estimators, robustness
under departures from the Gamma model, and extensions to other
distributional families with analogous sum-proportion decompositions.
}

%\clearpage

% \paragraph*{Disclosure statement}
% There are no conflicts of interest to disclose.

\section*{Declarations}

\paragraph*{Acknowledgements}
The authors would like to thank the editor and anonymous referees for their useful comments, which helped us improve the manuscript. The research was supported in part by CNPq and CAPES grants from the Brazilian government.

\paragraph*{Ethics Approval}
Not applicable.

% \paragraph*{Funding Declaration}
% This study was financed in part by CAPES (Finance Code 001). 

% \paragraph*{Disclosure statement}
% There are no conflicts of interest to disclose.

% \paragraph*{Data availability}
% The data supporting the findings of this study are publicly available from the Our World in Data repository at \url{https://ourworldindata.org/grapher/gdp-per-capita-worldbank} (accessed in June 2026).

% \paragraph*{Author Contributions Statement}
% All authors contributed equally to the conception and design of the study, data analysis, interpretation of the results, manuscript preparation, and revision. All authors read and approved the final version of the manuscript.

%%%%%%%%%%%%%%%%%%%%%%%%%%%%%%%%%%%%%%%%%%%%%%%%%%%%%%%%%%%%%

{%\small
\bibliographystyle{apalike}
%	\bibliography{BibTexFileName}

}

\end{document}